\documentclass[11pt,a4paper]{article}

\usepackage{amsmath,amsthm,amsfonts,amssymb}
\usepackage{multirow}
\usepackage{graphicx,color,filecontents}

\newtheorem{theorem}{Theorem}

\newtheorem{proposition}[theorem]{Proposition}
\theoremstyle{definition}

\newtheorem{definition}[theorem]{Definition}

\newcommand{\prob}{\mathbb{P}}

\renewcommand{\epsilon}{\varepsilon}
\renewcommand{\leq}{\leqslant}
\renewcommand{\geq}{\geqslant}
\newcommand{\fGr}[1][r]{{f(G;{#1})}}

\def\numleaves{\ell}
\def\ressize{m} 
\def\ampfactor{k}            
\def\superstar{S_{\numleaves,\ressize}^\ampfactor}      
\def\onevarsuperstar{S_{\numleaves,\ressize(\numleaves)}^\ampfactor}      
\def\onevarfivesuperstar{S_{\numleaves,\ressize(\numleaves)}^5}      
\def\samefivesuperstar{S_{\numleaves,\numleaves}^5}      
      
\def\numvertices{n}
\def\fixprob{f}
\def\eventR{R}
\def\eventF{F}
\def\probNoR{p(\numleaves,r)}
\def\probF{q(\numleaves,r)}

\title {On the fixation probability of superstars}

\author{Josep D{\'i}az%
        \thanks{\protect\raggedright
                Departament de Llenguatges i Sistemes Inform{\'a}tics,
                Universitat Polit{\'e}cnica de Catalunya, Spain. 
                Email: \{\texttt{diaz},
                            \texttt{mjserna}\}\texttt{@lsi.upc.edu}\,.},
        Leslie Ann Goldberg%
        \thanks{\protect\raggedright
                Department of Computer Science, University of Liverpool,
                UK.  Email: \{\texttt{L.A.Goldberg},
                \texttt{David.Richerby}\}\texttt{@liverpool.ac.uk}.
                Supported by EPSRC grant EP/I011528/1
                    \emph{Computational Counting}.},
        George B.~Mertzios%
        \thanks{\protect\raggedright
                School of Engineering and Computing Sciences,
                Durham University, UK.
                Email: \texttt{george.mertzios@durham.ac.uk}\,.},\\
        David Richerby\footnotemark[2], 
        Maria Serna\footnotemark[1]\ \ and 
        Paul G.~Spirakis%
        \thanks{\protect\raggedright
                Department of Computer Engineering and Informatics,
                University of Patras, Greece.
                Email: \texttt{spirakis@cti.gr}\,.}}

\date{\vspace{-0.8cm}}

\begin{document}
\maketitle{}

\begin{abstract}
    \noindent 
The Moran process models the spread of genetic mutations through a
population.  A mutant with relative fitness~$r$ is introduced into
a population and the system evolves, either reaching fixation (in
which every individual is a mutant) or extinction (in which none
is).  In a widely cited paper (\emph{Nature}, 2005), Lieberman,
Hauert and Nowak generalize the model to populations on the
vertices of graphs.  They describe a class of
graphs (called ``superstars''),  with a parameter~$k$.  Superstars
are designed
to have an increasing fixation probability as $k$ increases.
They state that   the probability of
fixation tends to $1-r^{-k}$ as graphs get larger but
we show that this claim
is untrue as stated. 
Specifically, for $k=5$, we show that the true fixation
probability (in the limit, as graphs get larger)
is at most $1-1/j(r)$ where $j(r)=\Theta(r^4)$,
contrary to the claimed result.  
We do believe that the qualitative claim of Lieberman et al.\ 
--- that the fixation probability of superstars tends to~$1$ as $k$~increases
--- is correct, and that it can probably be proved along the lines of their sketch.
We were able to run larger computer simulations than the ones
presented in their paper. However, simulations on graphs of around $40,000$
vertices do not support their claim. Perhaps these graphs are too small to exhibit
the limiting behaviour.
\end{abstract}

\noindent
{\bf Key words:} evolutionary dynamics, Moran process, fixation
    probability
 
\section{Introduction}
\label{sec:Intro}

The Moran process~\cite{Mor1958:Process} is a simple, discrete-time
model of the spread of genetic mutations through a finite population.
Individuals that do not possess the mutation have ``fitness''~1 and
mutants have fitness $r>0$.  At each time step, an individual is
selected, with probability proportional to its fitness, to reproduce.
A second individual is chosen uniformly at random, without
regard to fitness, and is replaced with a copy of
the reproducer.  Since the reproducer is chosen with probability
proportional to its fitness, the case $r>1$ corresponds to an
advantageous mutation.  With probability~1, the population will reach
one of two states, after which no further change is possible: the
population will consist entirely of mutants or of non-mutants.  These
scenarios are referred to as \emph{fixation} and \emph{extinction},
respectively.

Lieberman, Hauert and Nowak extend the model by structuring the
population on the vertices of a fixed directed
graph~\cite{LHN2005:Superstars}.  Each vertex corresponds to exactly
one individual.  In each time step of this generalized Moran process,
the reproducer is chosen as before: an individual is selected, with
probability proportional to its fitness.  Then a second
individual is selected uniformly at random from the set of
out-neighbours of the reproducer.  Once again, the second
individual is replaced with a copy of the reproducer.  The original
Moran process corresponds to the special case of the extended process
in which the graph is a complete graph (one with edges between all
pairs of individuals).

In this paper, we study the model of Lieberman, Hauert and Nowak.  It
is referred to as an \emph{invasion process} because an individual
duplicates and then replaces another.  This is in contrast to the
\emph{voter model}, which is another generalization of the Moran
process in which individuals first die, and are then replaced.  There
is much work on voter-model variants of the Moran process: see for
example~\cite{Mar1974:GeneFreq}.  In general, voter models and
invasion process behave differently~\cite{ARS2006:EvDyn}.
 
Given a graph $G$, we can ask what is the probability that a mutant
with fitness $r$ reaches fixation in the invasion process and we
denote this probability by $\fGr$.  It is easy to see that the number
of mutants in the original Moran process behaves as a random walk on
the integers with bias $r$ to the right and with absorbing barriers at
0 and~$N$, where $N$ is the population size.  Hence, as $N\to\infty$,
the fixation probability tends to $1-\frac{1}{r}$.  The generalized
Moran process can have a higher fixation probability.  For example, on
the complete bipartite graph $K_{1,N-1}$, the fixation probability
tends to $1-\frac{1}{r^2}$ as $N$ tends to infinity (see, for example,
Broom and Rycht\'a\v r's calculation \cite{Broom} of the exact
fixation probability, as a function of~$r$ and~$N$).
 
\subsection{Families of graphs with high fixation probability}

Lieberman \emph{et al.}\@ \cite{LHN2005:Superstars} introduce three
classes of graphs, which they call funnels, metafunnels and
superstars.  Superstars will be defined formally in
Section~\ref{sec:Bound}. An example is given in Figure~\ref{fig:superstar}.  
Funnels, metafunnels and superstars 
are essentially layered graphs, with
the addition of ``positive feedback loops''$\!$, and they have a
parameter~$k$ that corresponds to the number of layers.  
Lieberman, Hauert and Nowak claim
that, for fixed~$r>1$, for sufficiently large graphs in these classes,
the fixation probability tends to $1-r^{-k}$.  This is stated as
\cite[Theorem~3]{LHN2005:Superstars} for superstars and a  
proof sketch is given. 
Hauert states \cite[Equation
  (5)]{Hauert} that the same limiting fixation probability (and
presumably the same argument) also applies to funnels.  Lieberman \emph{et
al.}\@ conclude \cite{LHN2005:Superstars} that funnels, metafunnels
and superstars ``have the
amazing property that, for large $N$ [the number of vertices in the
  graph], the fixation probability of any advantageous mutant
converges to one. [...]  Hence, these population structures guarantee
fixation of advantageous mutants, however small their selective
advantage.''

The claimed limiting fixation probability of $1-r^{-k}$ is cited
frequently in the literature (see, for example, \cite[Equation
  (2)]{Barbosa1}, \cite[Equation (4)]{Barbosa2}, the survey paper
\cite[Equation (6)]{Shakarian2011} and the references therein).  We
prove that this limiting fixation probability is incorrect for $k=5$,
demonstrating that the 
proof sketch cannot be made rigorous, at least for the exact claim that they make.

On the other hand, superstars do seem to be well-designed to amplify selection.
Informally, the chains in these graphs (such as the chain $c_{1,1}, c_{1,2}, c_{1,3}$
in Figure~\ref{fig:superstar}) seem to be a good mechanism for amplifying
the fitness of a mutant, and the trade-off between the high out-degree of the centre
vertex and the lower in-degree seems to be a useful feature.

We have investigated the fixation probability of superstars via computer simulation.
Before discussing our  proof, and the result of these
simulations, we give a brief survey of
the relevant literature.  Lieberman \emph{et al.}\@ \cite{LHN2005:Superstars}
simulated the fixation probability of superstars for the special cases
when $r=1.1$ and $k=3$ and $k=4$ on graphs of around 10,000 vertices.
Unfortunately, these particular values are too small to give
evidence of their general claim.

Funnels and metafunnels are not very amenable to simulation since the
number of vertices is exponential in the relevant parameters.  We are
not aware of any published justification for the claim for
metafunnels but there has been some simulation work relevant to
funnels.  Barbosa \emph{et al.}\@ \cite{Barbosa1} have found the
fixation probability to be close to $1-r^{-3}$ for funnels of up to
around 1,600 vertices for the special cases $k=3$ and $r=1.1$ and
$r=2$.  Motivated by the claimed fixation probability for funnels,
their objective was to see whether similar phenomena occur for similar
randomly generated layered graphs, which they argue are more like
``naturally occurring population structures'' than are funnels,
metafunnels and superstars.  They found that the fixation
probabilities for $r=1.1$ and $r=2$ on these randomly generated graphs
with $k=5$ or
$k=10$ generally exceed the value of $1-\tfrac{1}{r}$ that would be
seen in an unstructured population but are substantially lower than
$1-r^{-k}$.  These experiments do not apply directly to funnels (and
it may be that the graphs that they considered were too small to
demonstrate the limit behaviour) but, in any case, their experiments
do not give evidence in favour of the fixation probability claimed by
Lieberman \emph{et al.}\@ \cite{LHN2005:Superstars}.

For small graphs it is possible to calculate exact fixation
probabilities by solving a linear system. If the graph has~$n$
vertices, then the Moran process has $2^n$~states, so there are
$2^n$~equations in the linear system.  Computationally, solving such a
system is not feasible, apart from for tiny graphs.  A significant
improvement was introduced by Houchmandzadeh and
Vallade~\cite{HV2011:FixProb}, who present a new method for
calculating fixation probabilities by solving differential equations.
The relevant equation \cite[Equation 23]{HV2011:FixProb} has a
variable~$z_i$ for each vertex~$i$, so there are~$n$ variables in all.
A further improvement is given: if the vertices in the graph can be
partitioned into equivalence classes such that all of the vertices in
a given equivalence class have exactly the same set of in-neighbours
and the same set of out-neighbours then these vertices can share a
variable (in the terminology of~\cite{HV2011:FixProb}, they can be
viewed as a single ``island''). Thus, the fixation probability can be
calculated by solving a differential equation in which the number of
variables equals the number of equivalence classes. The paper
\cite{HV2011:FixProb} also offers a method for approximately solving
the relevant differential equations. This seems to work well in
practice, though the approximation is difficult to analyse and there
are currently no known results guaranteeing how close the approximate
value will be to the actual fixation probability.

\subsection{Outline of the paper}

In  Section~\ref{sec:Bound}, we prove that the fixation probability for sufficiently
large parameter-5 superstars cannot exceed $1-\frac{r+1}{2r^5+r+1}$,
which is clearly bounded below $1-r^{-5}$ for all  sufficiently large $r$
(in particular, for $r\geq 1.42$).  This proof is fully rigorous, though
we use a computer algebra package to invert a $31\times 31$ symbolic
matrix.  Thus, we show that \cite[Theorem 3]{LHN2005:Superstars} is incorrect
as stated (though something very similar may well be true).

Section~\ref{sec:Simulation} presents simulation results 
on graphs of around $40,000$ vertices.
These simulations do not support the claim that the fixation probability 
is $1-r^{-k}$, or that this probability increases as $k$ increases.
However, it may be that $40,000$ vertices is not enough to exhibit the true limiting behaviour.

\section{An upper bound for $k=5$}
\label{sec:Bound}

The superstars of Lieberman \emph{et al.}\@ are defined as follows.
A superstar $\superstar$ has a \emph{centre}
vertex~$v$, and $\numleaves$ disjoint subgraphs called \emph{leaves}.
Each leaf consists of a \emph{reservoir} of~$\ressize$ vertices, together
with a chain of length $\ampfactor-2$.  There are edges from the
centre to the reservoir vertices, from the reservoir vertices to
the start of the chain, and from the end of the chain back to the
centre.  The formal definition follows, where $[n]$ denotes the set
$\{1, \dots, n\}$.

\begin{figure}
    \begin{center}
  \begin{picture}(0,0)%
\includegraphics{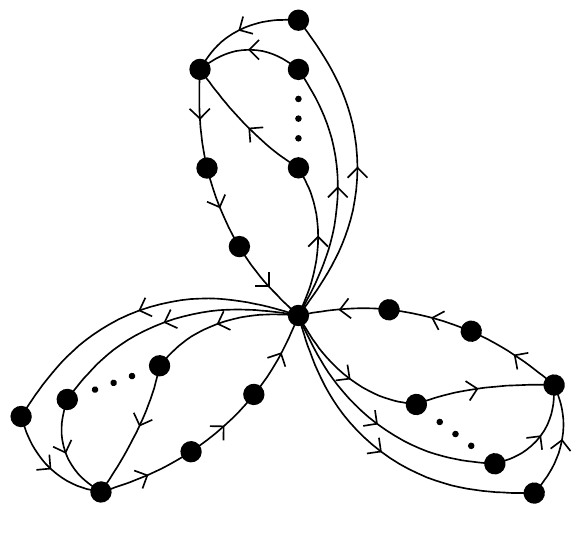}%
\end{picture}%
\setlength{\unitlength}{4144sp}%
\begingroup\makeatletter\ifx\SetFigFont\undefined%
\gdef\SetFigFont#1#2#3#4#5{%
  \reset@font\fontsize{#1}{#2pt}%
  \fontfamily{#3}\fontseries{#4}\fontshape{#5}%
  \selectfont}%
\fi\endgroup%
\begin{picture}(2640,2490)(3136,-4711)
\put(3601,-4651){\makebox(0,0)[b]{\smash{{\SetFigFont{10}{12.0}{\familydefault}{\mddefault}{\updefault}{\color[rgb]{0,0,0}$c_{3,1}$}%
}}}}
\put(3961,-4471){\makebox(0,0)[lb]{\smash{{\SetFigFont{10}{12.0}{\familydefault}{\mddefault}{\updefault}{\color[rgb]{0,0,0}$c_{3,2}$}%
}}}}
\put(4276,-4201){\makebox(0,0)[lb]{\smash{{\SetFigFont{10}{12.0}{\familydefault}{\mddefault}{\updefault}{\color[rgb]{0,0,0}$c_{3,3}$}%
}}}}
\put(3466,-4201){\makebox(0,0)[lb]{\smash{{\SetFigFont{10}{12.0}{\familydefault}{\mddefault}{\updefault}{\color[rgb]{0,0,0}$x_{3,2}$}%
}}}}
\put(3151,-4156){\makebox(0,0)[rb]{\smash{{\SetFigFont{10}{12.0}{\familydefault}{\mddefault}{\updefault}{\color[rgb]{0,0,0}$x_{3,1}$}%
}}}}
\put(3961,-3886){\makebox(0,0)[lb]{\smash{{\SetFigFont{10}{12.0}{\familydefault}{\mddefault}{\updefault}{\color[rgb]{0,0,0}$x_{3,m}$}%
}}}}
\put(4411,-3796){\makebox(0,0)[rb]{\smash{{\SetFigFont{10}{12.0}{\familydefault}{\mddefault}{\updefault}{\color[rgb]{0,0,0}$v$}%
}}}}
\put(4141,-3391){\makebox(0,0)[rb]{\smash{{\SetFigFont{10}{12.0}{\familydefault}{\mddefault}{\updefault}{\color[rgb]{0,0,0}$c_{1,3}$}%
}}}}
\put(4006,-3031){\makebox(0,0)[rb]{\smash{{\SetFigFont{10}{12.0}{\familydefault}{\mddefault}{\updefault}{\color[rgb]{0,0,0}$c_{1,2}$}%
}}}}
\put(3961,-2581){\makebox(0,0)[rb]{\smash{{\SetFigFont{10}{12.0}{\familydefault}{\mddefault}{\updefault}{\color[rgb]{0,0,0}$c_{1,1}$}%
}}}}
\put(4591,-2356){\makebox(0,0)[lb]{\smash{{\SetFigFont{10}{12.0}{\familydefault}{\mddefault}{\updefault}{\color[rgb]{0,0,0}$x_{1,1}$}%
}}}}
\put(4681,-2581){\makebox(0,0)[lb]{\smash{{\SetFigFont{10}{12.0}{\familydefault}{\mddefault}{\updefault}{\color[rgb]{0,0,0}$x_{1,2}$}%
}}}}
\put(4861,-3031){\makebox(0,0)[lb]{\smash{{\SetFigFont{10}{12.0}{\familydefault}{\mddefault}{\updefault}{\color[rgb]{0,0,0}$x_{1,m}$}%
}}}}
\put(4951,-3526){\makebox(0,0)[b]{\smash{{\SetFigFont{10}{12.0}{\familydefault}{\mddefault}{\updefault}{\color[rgb]{0,0,0}$c_{2,3}$}%
}}}}
\put(5266,-3616){\makebox(0,0)[lb]{\smash{{\SetFigFont{10}{12.0}{\familydefault}{\mddefault}{\updefault}{\color[rgb]{0,0,0}$c_{2,2}$}%
}}}}
\put(5041,-3931){\makebox(0,0)[b]{\smash{{\SetFigFont{10}{12.0}{\familydefault}{\mddefault}{\updefault}{\color[rgb]{0,0,0}$x_{2,m}$}%
}}}}
\put(5761,-4021){\makebox(0,0)[lb]{\smash{{\SetFigFont{10}{12.0}{\familydefault}{\mddefault}{\updefault}{\color[rgb]{0,0,0}$c_{2,1}$}%
}}}}
\put(5401,-4201){\makebox(0,0)[b]{\smash{{\SetFigFont{10}{12.0}{\familydefault}{\mddefault}{\updefault}{\color[rgb]{0,0,0}$x_{2,2}$}%
}}}}
\put(5671,-4516){\makebox(0,0)[lb]{\smash{{\SetFigFont{10}{12.0}{\familydefault}{\mddefault}{\updefault}{\color[rgb]{0,0,0}$x_{2,1}$}%
}}}}
\end{picture}%
    \end{center}
    \vspace{-5mm}
    \caption{The superstar $S^5_{3,m}$.  The three leaves each have
      $m$ reservoir vertices and a chain of length $5-2=3$ between
      them and the centre, $v$.}
    \label{fig:superstar}
\end{figure}
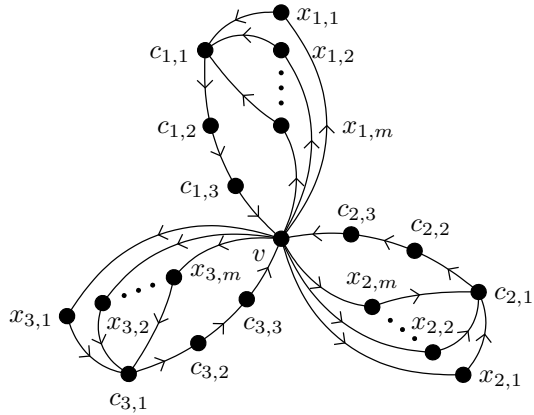
  
\begin{definition}[Lieberman \emph{et al.}~\cite{LHN2005:Superstars}]
    Define
    \begin{gather*}
        V = \{v\} \cup \{ x_{i,j} \mid  i \in[\numleaves],j\in[\ressize]\} \cup
            \{c_{i, j} \mid i \in [\numleaves], j \in [\ampfactor-2]\}, \mbox{and}\\
\begin{aligned}
E = \{(v,x_{i,j}),(x_{i,j},c_{i,1}) \mid i\in[\numleaves], j \in[\ressize] \} &\cup 
\{(c_{i,j},c_{i,j+1}) \mid i\in[\numleaves], j\in[\ampfactor-3]\} 
\\ &\cup 
\{(c_{i,\ampfactor-2},v) \mid i\in[\numleaves]\}.
    \end{aligned}
    \end{gather*} 
    The graph $\superstar=(V,E)$ is a \emph{parameter-$k$ superstar}
    with $\numleaves$ leaves and reservoir size~$\ressize$.  We
    use~$\numvertices$ to denote $|V| = 1+ \numleaves (\ressize +
    \ampfactor-2)$.
\end{definition}

Figure~\ref{fig:superstar} shows the parameter-5 superstar
$S^5_{3,m}$.  The parameter $k$ is sometimes referred to as the
``amplification factor''.

Lieberman \emph{et al.}\@ state the following proposition (which turns
out to be incorrect --- see Theorem~\ref{thm:kfive} below).
 
\begin{proposition}[{Stated as \cite[Theorem 3]{LHN2005:Superstars}}]
\label{thm:original}
    \begin{equation*}
      \lim_{\numleaves,\ressize\rightarrow\infty} \fixprob({\superstar};r)
          = \frac{1-r^{-\ampfactor}}{1-r^{-\ampfactor \numvertices}}.
    \end{equation*}
\end{proposition}

The statement of Proposition~\ref{thm:original} is not sufficiently
precise because the two variables $\numleaves$ and $\ressize$ are
simultaneously taken to infinity without regard to the relative rates
at which they tend to infinity.  Nevertheless, we can make sense of
the proposition by regarding $\ressize$ as a function of~$\numleaves$.
We require that $m(\ell)=\omega(1)$, i.e., that the function $m(\ell)$
is an increasing function of $\ell$ that grows without bound.
For example Nowak~\cite{Nowak2006:EvDyn} considers
$\ressize=\numleaves$.  Since we are only interested in $r>1$, we can
also simplify the expression, using the fact that the denominator
tends to~$1$.

\begin{proposition}[{The $r>1$ case of \cite[Theorem~3]{LHN2005:Superstars}}]
\label{thm:clarified}
    Suppose $r>1$ and $\ressize(\numleaves)=\omega(1)$.  Then
    \begin{equation*}
        \lim_{\numleaves\rightarrow\infty} \fixprob({\onevarsuperstar};r)
            = {1- r^{-\ampfactor}}.
    \end{equation*}
\end{proposition}

Lieberman \emph{et al.}\@ give a brief sketch of a proposed proof of
Proposition~\ref{thm:clarified}. 
However, we now show that this sketch cannot be made rigorous for the proposition as stated.
We do this by choosing a fixed   value of $\ampfactor$ (specifically,
$\ampfactor=5$) and  showing that Proposition~\ref{thm:clarified} is
false for this value of~$\ampfactor$.  Specifically, we show the
following:

\begin{theorem} \label{thm:kfive}
    Let $\ressize(\numleaves)$ be any function which is $\omega(1)$.
    Let $j(r) = \frac{2r^5+r+1}{r+1}$.  For any $r>1$, if
    $\lim_{\numleaves\rightarrow\infty}
    \fixprob({\onevarfivesuperstar};r)$ exists, then
    \begin{equation*}
        \lim_{\numleaves\rightarrow\infty} \fixprob({\onevarfivesuperstar};r)
            \leq 1 - \frac{1}{j(r)}.
    \end{equation*}
\end{theorem}

Note that Theorem~\ref{thm:kfive} applies for any function
$\ressize(\numleaves)=\omega(1)$.  In particular, it shows that, for
all~$r>1$, if $\lim_{\numleaves\rightarrow\infty}
\fixprob({\samefivesuperstar};r)$ exists then
\begin{equation*}
    \lim_{\numleaves\rightarrow\infty} \fixprob({\samefivesuperstar};r)
        \leq 1-\frac{1}{j(r)},
\end{equation*}
whereas Proposition~\ref{thm:clarified} would give the contrary
conclusion
\begin{equation*}
    \lim_{\numleaves\rightarrow\infty} \fixprob({\samefivesuperstar};r)
        = 1- \frac{1}{r^{5}},
\end{equation*}
where $1- \frac{1}{r^{5}} > 1-\frac{1}{j(r)}$ for all sufficiently
large~$r$ (specifically, for $r\geq 1.42$) since $j(r) = \Theta(r^4)$.

\begin{proof}[Proof of Theorem~\ref{thm:kfive}]
    Let $\ressize(\numleaves)$ be any function which is $\omega(1)$.
    Consider the generalized Moran process on~$\onevarfivesuperstar$.  Let
    $\eventR$ be the event that the initial mutant is placed on a
    reservoir vertex and let $\eventF$ be the event that, at some time
    during the execution of the process, the centre vertex~$v$ is
    occupied by a mutant and is chosen for reproduction.  Let
    $\probNoR$ be the probability that $\eventR$ does not occur and
    let $\probF$ be the probability that $\eventF$ occurs, conditioned
    on the fact that event~$\eventR$ occurs.  Clearly,
    \begin{align*}
        \fixprob({\onevarfivesuperstar};r) \leq \prob[\eventF] 
            &\leq \probNoR + \probF \\
            &= \frac{1+ 3 \numleaves }{\numvertices} + \probF
             = \frac{1+ 3 \numleaves }
                    { 1+ \numleaves (\ressize(\numleaves) + 3)} + \probF.
    \end{align*}
    Let $h(r) = \lim_{\numleaves\rightarrow\infty} \probF$.  We will
    show that this limit exists for every $r>0$, and that $h(r) = 1 -
    \frac{1}{j(r)}$.  From the calculation above, it is clear that,
    for every $r>1$, if $ \lim_{\numleaves\rightarrow\infty}
    \fixprob({\onevarfivesuperstar};r)$ exists then
    \begin{equation*}
        \lim_{\numleaves\rightarrow\infty} \fixprob({\onevarfivesuperstar};r)
            \leq \lim_{\numleaves\rightarrow\infty} \probF = h(r).
    \end{equation*}
 
    In fact, the value $\probF$ is a rational function in the
    variables~$\numleaves$, $\ressize(\numleaves)$ and~$r$. This
    rational function can be calculated by solving a linear system.
    We solved this linear system using Mathematica --- the
    corresponding Mathematica program is in
    Appendix~\ref{app:Mathematica}.  The program consists of three
    main parts.  The first block of code defines useful constants, the
    bulk of the file defines the system of linear equations and the
    last four blocks solve the system for all variables and extract
    the solution of interest.

    In the Mathematica program, $V$ denotes the vertex~$v$, $X$
    denotes the reservoir vertex~$x_{i,j}$ in which the initial mutant
    is placed, and $O$, $P$ and $Q$ represent the vertices in the
    corresponding chain ($c_{i,1}$, $c_{i,2}$ and $c_{i,3}$,
    respectively).  Let $\Psi = \{V, X, O, P, Q\}$.  If we start the
    generalized Moran process from the state in which vertex~$X$ is
    occupied by a mutant, and no other vertices are occupied by
    mutants, then no vertices outside~$\Psi$ can be occupied by
    mutants until event $\eventF$ occurs.  In the program, $L$ is a
    variable representing the quantity~$\numleaves$ and $M$ is a
    variable representing the quantity~$\ressize(\numleaves)$.
    Let~$\Omega$ be the state space of the generalized Moran process,
    which contains one state for each subset of $\Psi$.  The state
    corresponding to subset $S \in \Omega$ is the state in which
    the vertices in~$S$ are occupied by mutants and no other vertices are
    occupied by mutants.  We use the
    program variable $FS$ to denote the probability that
    event~$\eventF$ occurs, starting from state~$S$.

    For each state $S$, $EQS$ is a
    linear equation relating $FS$ to to the other variables in $\{
    F{S'} \mid S' \in \Omega\}$.  The linear equations can be derived
    by considering the transitions of the system.  To aid the reader,
    we give an example.  Consider the state $XO$ in which vertices $X$
    and $O$ are occupied by mutants.  From this state, three
    transitions are possible.  (We write $W$ for the total fitness of
    vertices in the state under consideration.)
    \begin{itemize}
    \item With probability $\tfrac{r}{W}$, vertex $O$ is chosen for
        reproduction. Vertex~$P$ becomes a mutant so the new state is
        $XOP$.
    \item With probability $\tfrac{1}{W}\times \tfrac{1}{ L M}$,
        vertex $V$ is chosen for reproduction.  From among its $L
        M$ neighbours, it chooses vertex $X$ to update (removing the
        mutant from vertex~$X$), so the new state is $O$.
    \item With probability $\frac{M-1}{W}$, one of the vertices in
        $\{x_{i,j} \mid j\in[\ressize(\numleaves)]\}\setminus X$ is
        chosen for reproduction, removing the mutant from vertex~$O$,
        so the new state is $X$.
\end{itemize}
    Thus, we have the equality
    \begin{equation*}
        FXO = \frac{\tfrac{r}{W} FXOP + \tfrac{1}{W}\tfrac{1}{L M} FO
                        + \frac{M-1}{W} FX}
                   {\tfrac{r}{W} + \tfrac{1}{W} \tfrac{1}{L M} 
                        + \frac{M-1}{W} }.
    \end{equation*}

    This equality (which we called $EQXO$) is included in the linear
    system constructed in the Mathematica program (except that we
    normalized by multiplying the numerator and denominator by~$W$).
    The constant $DXO$ is defined to stand for the denominator of this
    expression to enhance readability.  The constants \emph{XonO},
    \emph{XoffO} and so on refer to the probabilities that,
    respectively, the vertex $O$ is made a mutant or a non-mutant
    (``switched on or off'') by $X$ (again, normalized by multiplying
    by $W$).

    We similarly derive an equation $EQS$ for every non-empty state
    $S\in \Omega$.  Clearly, if $S$ is the state in which no vertices
    are mutants then $FS=0$, so we can account for this directly in
    the other equations.  The system therefore consists of $31$
    equations in $31$ variables with one variable $FS$ for each
    non-empty state~$S\in\Omega$.  The desired quantity~$\probF$ is
    equal to~$FX$, which can therefore be calculated by (symbolically)
    solving the linear system.
 
    The solution for~$FX$ is a rational function in~$L$, $M$ and~$r$.
    The numerator of this rational function can be written as
    $\sum_{i=0}^{19}\sum_{j=0}^{19} c_{i,j}(r) L^i {M}^j$.  We say
    that the term $c_{i,j}(r) L^i M^j$ is \emph{dominated} by the term
    $c_{i',j'}(r) L^{i'} M^{j'}$ if $c_{i',j'}\neq 0$, $i\leq i'$,
    $j\leq j'$ and $i+j<i'+j'$.  The sum of the undominated terms in
    the numerator  is
    \begin{equation*}
        2 r^5(1+r) L^{14} M^{14} {(L+M)}^5.
    \end{equation*}
    Similarly, the sum of the undominated terms in the denominator is
    \begin{equation*}
       (1+ 2 r + r^2 + 2r^5  + 2 r^6) L^{14} M^{14} {(L+M)^5}.
    \end{equation*}

    Thus, for any fixed $r$,
    \begin{equation*}
        \lim_{\numleaves\rightarrow\infty} \probF 
            = \frac{2 r^5(1+r)}{1+ 2 r + r^2 + 2r^5  + 2 r^6}
            = \frac{2 r^5 }{1+  r + 2r^5 }
            = 1 - \frac{1+r}{1+r+2r^5}.
    \end{equation*}
    Since $j(r) = \frac{2r^5+r+1}{r+1}$, we have
    $\lim_{\numleaves\rightarrow\infty} \probF = 1 - \frac{1}{j(r)}$.
\end{proof} 

 \section{Simulations on superstars}
\label{sec:Simulation}

We simulated the generalized Moran process on superstars with $\ell=m=200$ and for
$k\in\{3, 4, 5, 6, 7, 12\}$ and $r\in\{1.1, 2, 3, 5, 10, 50\}$.  Thus,
the size of the graphs ranges from approximately 40,000 to
approximately 42,000 vertices.  For each choice of parameters, we ran
2,500 simulations for $r\leq 5$ and 10,000 for $r\geq 10$.  The
results are presented in Table~\ref{table:Superstar-sim} and
Figure~\ref{fig:Superstar-sim}.

\begin{table}
\newcommand{\ci}[1]{\tiny [#1]}
\begin{center}
\begin{tabular}{rccccccc}
    \hline
    {} & $r=1.1$ & $r=2$ & $r=3$ & $r=5$ & $r=10$ & $r=50$ \\
    \hline
    \multirow{2}{*}{$k=3$}
           & 0.248 & 0.872 & 0.951 & 0.980 & 0.994 & 0.995 \\[-1ex]
           & \ci{0.225, 0.273} & \ci{0.852, 0.889} & \ci{0.938, 0.962}
           & \ci{0.971, 0.987} & \ci{0.991, 0.995} & \ci{0.993, 0.997} \\[1ex]
    \multirow{2}{*}{$k=4$}
           & 0.292 & 0.923 & 0.979 & 0.986 & 0.991 & 0.995 \\[-1ex]
           & \ci{0.267, 0.318} & \ci{0.906, 0.937} & \ci{0.969, 0.986}
           & \ci{0.977, 0.991} & \ci{0.988, 0.994} & \ci{0.993, 0.997} \\[1ex]
    \multirow{2}{*}{$k=5$}
           & 0.333 & 0.938 & 0.978 & 0.989 & 0.990 & 0.995 \\[-1ex]
           & \ci{0.307, 0.360} & \ci{0.923, 0.950} & \ci{0.969, 0.985}
           & \ci{0.981, 0.994} & \ci{0.987, 0.993} & \ci{0.993, 0.997} \\[1ex]
    \multirow{2}{*}{$k=6$}
           & 0.362 & 0.934 & 0.970 & 0.983 & 0.987 & 0.996 \\[-1ex]
           & \ci{0.336, 0.389} & \ci{0.918, 0.946} & \ci{0.959, 0.978}
           & \ci{0.974, 0.989} & \ci{0.984, 0.990} & \ci{0.994, 0.998} \\[1ex]
    \multirow{2}{*}{$k=7$}
           & 0.374 & 0.948 & 0.972 & 0.978 & 0.986 & 0.996 \\[-1ex]
           & \ci{0.347, 0.402} & \ci{0.934, 0.960} & \ci{0.962, 0.980}
           & \ci{0.969, 0.985} & \ci{0.982, 0.989} & \ci{0.996, 0.998} \\[1ex]
    \multirow{2}{*}{$k=12$}
           & 0.419 & 0.928 & 0.953 & 0.962 & 0.982 & 0.994 \\[-1ex]
           & \ci{0.391, 0.447} & \ci{0.913, 0.942} & \ci{0.939, 0.963}
           & \ci{0.950, 0.972} & \ci{0.978, 0.985} & \ci{0.992, 0.996} \\
    \hline
\end{tabular}
\caption{Superstar fixation probabilities obtained by simulation.  The
  range in small type is the 99.5\% confidence interval, which is not
  symmetric about the sample mean.  Sample size is 2,500 simulations
  for $r\leq 5$ and 10,000 for $r\geq 10$.}
\label{table:Superstar-sim}
\end{center}
\end{table}

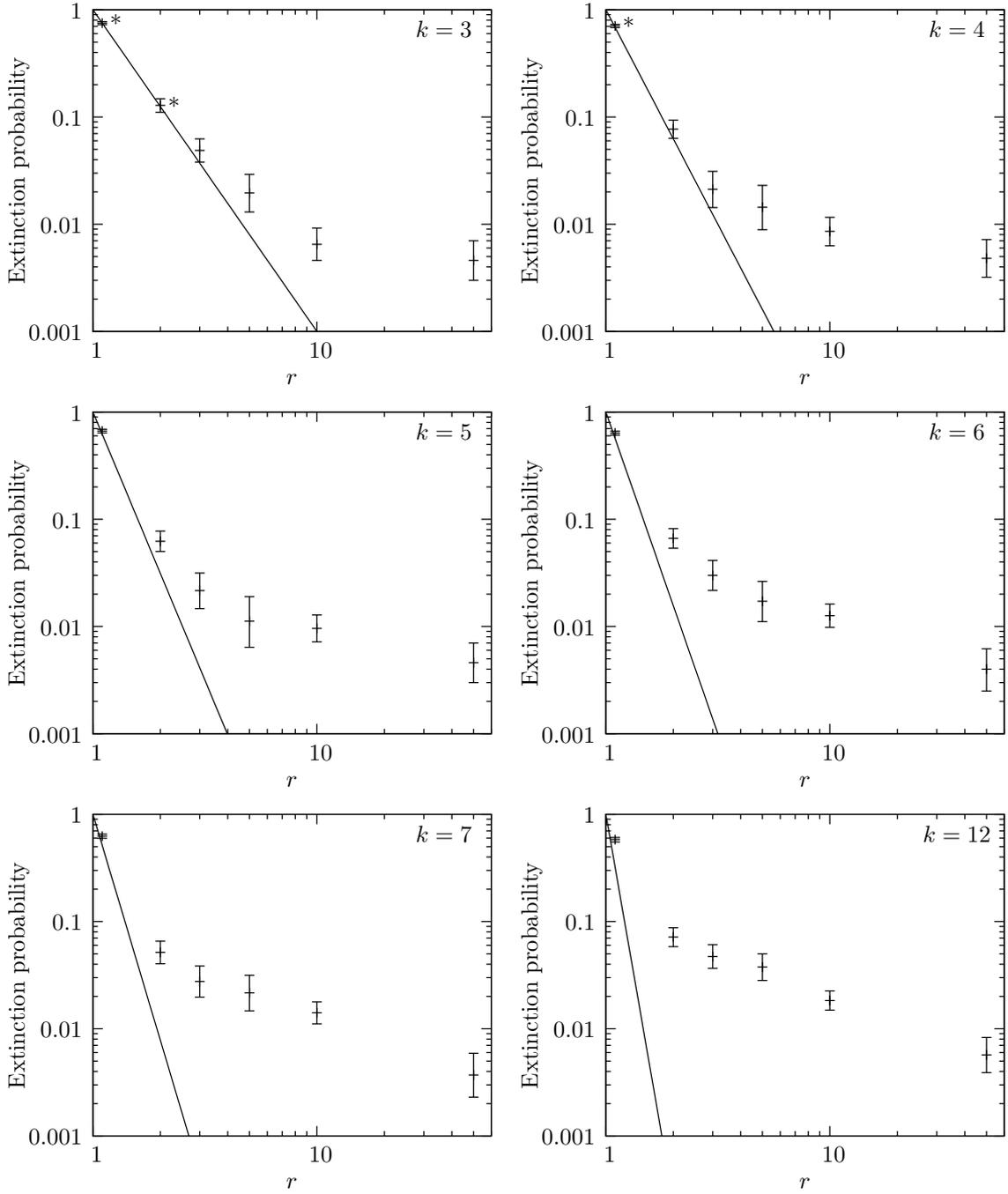
\begin{figure}
\hspace{-1.5cm}
\begin{tabular}{ll}
\begin{picture}(0,0)%
\includegraphics{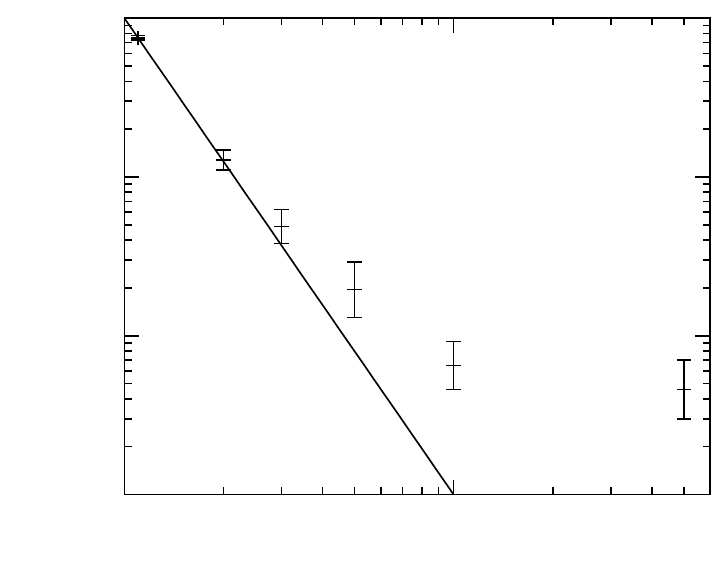}%
\end{picture}%
\setlength{\unitlength}{4144sp}%
\begingroup\makeatletter\ifx\SetFigFont\undefined%
\gdef\SetFigFont#1#2#3#4#5{%
  \reset@font\fontsize{#1}{#2pt}%
  \fontfamily{#3}\fontseries{#4}\fontshape{#5}%
  \selectfont}%
\fi\endgroup%
\begin{picture}(3257,2659)(1078,-2787)
\put(1575,-2440){\makebox(0,0)[rb]{\smash{{\SetFigFont{10}{12.0}{\familydefault}{\mddefault}{\updefault} 0.001}}}}
\put(1575,-1714){\makebox(0,0)[rb]{\smash{{\SetFigFont{10}{12.0}{\familydefault}{\mddefault}{\updefault} 0.01}}}}
\put(1575,-988){\makebox(0,0)[rb]{\smash{{\SetFigFont{10}{12.0}{\familydefault}{\mddefault}{\updefault} 0.1}}}}
\put(1575,-262){\makebox(0,0)[rb]{\smash{{\SetFigFont{10}{12.0}{\familydefault}{\mddefault}{\updefault} 1}}}}
\put(1646,-2559){\makebox(0,0)[b]{\smash{{\SetFigFont{10}{12.0}{\familydefault}{\mddefault}{\updefault} 1}}}}
\put(3151,-2559){\makebox(0,0)[b]{\smash{{\SetFigFont{10}{12.0}{\familydefault}{\mddefault}{\updefault} 10}}}}
\put(1189,-1292){\rotatebox{90.0}{\makebox(0,0)[b]{\smash{{\SetFigFont{10}{12.0}{\familydefault}{\mddefault}{\updefault}Extinction probability}}}}}
\put(2984,-2737){\makebox(0,0)[b]{\smash{{\SetFigFont{10}{12.0}{\familydefault}{\mddefault}{\updefault}$r$}}}}
\put(1708,-352){\makebox(0,0)[lb]{\smash{{\SetFigFont{10}{12.0}{\familydefault}{\mddefault}{\updefault} *}}}}
\put(2099,-909){\makebox(0,0)[lb]{\smash{{\SetFigFont{10}{12.0}{\familydefault}{\mddefault}{\updefault} *}}}}
\put(4005,-389){\makebox(0,0)[b]{\smash{{\SetFigFont{10}{12.0}{\familydefault}{\mddefault}{\updefault}$k=3$}}}}
\end{picture}%
&
\begin{picture}(0,0)%
\includegraphics{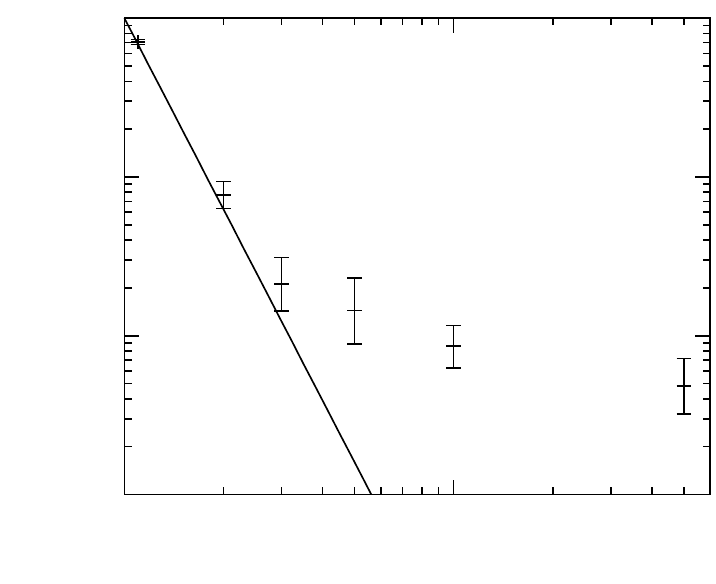}%
\end{picture}%
\setlength{\unitlength}{4144sp}%
\begingroup\makeatletter\ifx\SetFigFont\undefined%
\gdef\SetFigFont#1#2#3#4#5{%
  \reset@font\fontsize{#1}{#2pt}%
  \fontfamily{#3}\fontseries{#4}\fontshape{#5}%
  \selectfont}%
\fi\endgroup%
\begin{picture}(3257,2659)(1078,-2787)
\put(1575,-2440){\makebox(0,0)[rb]{\smash{{\SetFigFont{10}{12.0}{\familydefault}{\mddefault}{\updefault} 0.001}}}}
\put(1575,-1714){\makebox(0,0)[rb]{\smash{{\SetFigFont{10}{12.0}{\familydefault}{\mddefault}{\updefault} 0.01}}}}
\put(1575,-988){\makebox(0,0)[rb]{\smash{{\SetFigFont{10}{12.0}{\familydefault}{\mddefault}{\updefault} 0.1}}}}
\put(1575,-262){\makebox(0,0)[rb]{\smash{{\SetFigFont{10}{12.0}{\familydefault}{\mddefault}{\updefault} 1}}}}
\put(1646,-2559){\makebox(0,0)[b]{\smash{{\SetFigFont{10}{12.0}{\familydefault}{\mddefault}{\updefault} 1}}}}
\put(3151,-2559){\makebox(0,0)[b]{\smash{{\SetFigFont{10}{12.0}{\familydefault}{\mddefault}{\updefault} 10}}}}
\put(1189,-1292){\rotatebox{90.0}{\makebox(0,0)[b]{\smash{{\SetFigFont{10}{12.0}{\familydefault}{\mddefault}{\updefault}Extinction probability}}}}}
\put(2984,-2737){\makebox(0,0)[b]{\smash{{\SetFigFont{10}{12.0}{\familydefault}{\mddefault}{\updefault}$r$}}}}
\put(1708,-371){\makebox(0,0)[lb]{\smash{{\SetFigFont{10}{12.0}{\familydefault}{\mddefault}{\updefault} *}}}}
\put(4005,-389){\makebox(0,0)[b]{\smash{{\SetFigFont{10}{12.0}{\familydefault}{\mddefault}{\updefault}$k=4$}}}}
\end{picture}%
\\
\begin{picture}(0,0)%
\includegraphics{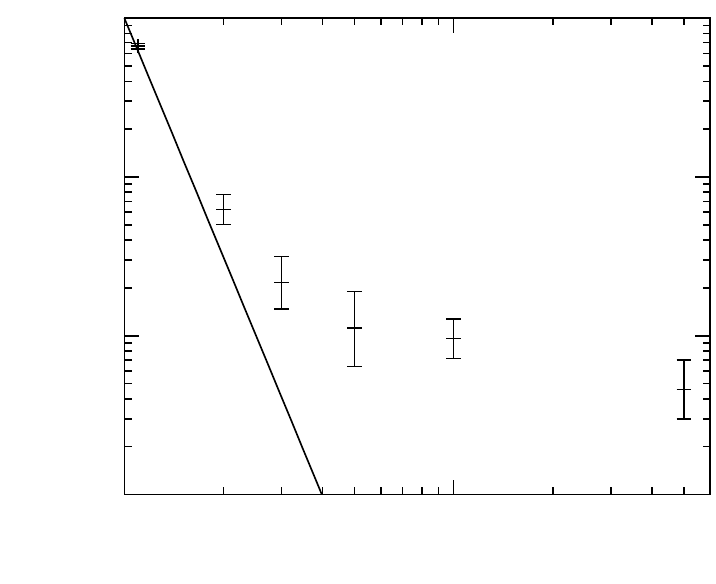}%
\end{picture}%
\setlength{\unitlength}{4144sp}%
\begingroup\makeatletter\ifx\SetFigFont\undefined%
\gdef\SetFigFont#1#2#3#4#5{%
  \reset@font\fontsize{#1}{#2pt}%
  \fontfamily{#3}\fontseries{#4}\fontshape{#5}%
  \selectfont}%
\fi\endgroup%
\begin{picture}(3257,2659)(1078,-2787)
\put(1575,-2440){\makebox(0,0)[rb]{\smash{{\SetFigFont{10}{12.0}{\familydefault}{\mddefault}{\updefault} 0.001}}}}
\put(1575,-1714){\makebox(0,0)[rb]{\smash{{\SetFigFont{10}{12.0}{\familydefault}{\mddefault}{\updefault} 0.01}}}}
\put(1575,-988){\makebox(0,0)[rb]{\smash{{\SetFigFont{10}{12.0}{\familydefault}{\mddefault}{\updefault} 0.1}}}}
\put(1575,-262){\makebox(0,0)[rb]{\smash{{\SetFigFont{10}{12.0}{\familydefault}{\mddefault}{\updefault} 1}}}}
\put(1646,-2559){\makebox(0,0)[b]{\smash{{\SetFigFont{10}{12.0}{\familydefault}{\mddefault}{\updefault} 1}}}}
\put(3151,-2559){\makebox(0,0)[b]{\smash{{\SetFigFont{10}{12.0}{\familydefault}{\mddefault}{\updefault} 10}}}}
\put(1189,-1292){\rotatebox{90.0}{\makebox(0,0)[b]{\smash{{\SetFigFont{10}{12.0}{\familydefault}{\mddefault}{\updefault}Extinction probability}}}}}
\put(2984,-2737){\makebox(0,0)[b]{\smash{{\SetFigFont{10}{12.0}{\familydefault}{\mddefault}{\updefault}$r$}}}}
\put(4005,-389){\makebox(0,0)[b]{\smash{{\SetFigFont{10}{12.0}{\familydefault}{\mddefault}{\updefault}$k=5$}}}}
\end{picture}%
&
\begin{picture}(0,0)%
\includegraphics{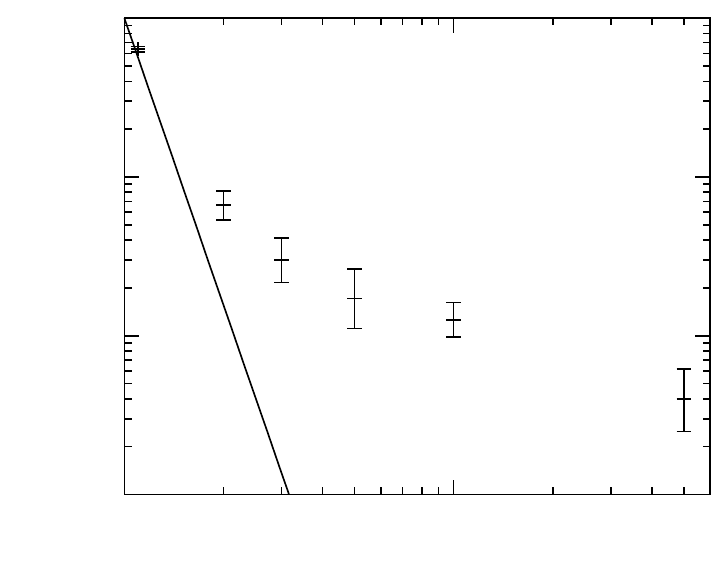}%
\end{picture}%
\setlength{\unitlength}{4144sp}%
\begingroup\makeatletter\ifx\SetFigFont\undefined%
\gdef\SetFigFont#1#2#3#4#5{%
  \reset@font\fontsize{#1}{#2pt}%
  \fontfamily{#3}\fontseries{#4}\fontshape{#5}%
  \selectfont}%
\fi\endgroup%
\begin{picture}(3257,2659)(1078,-2787)
\put(1575,-2440){\makebox(0,0)[rb]{\smash{{\SetFigFont{10}{12.0}{\familydefault}{\mddefault}{\updefault} 0.001}}}}
\put(1575,-1714){\makebox(0,0)[rb]{\smash{{\SetFigFont{10}{12.0}{\familydefault}{\mddefault}{\updefault} 0.01}}}}
\put(1575,-988){\makebox(0,0)[rb]{\smash{{\SetFigFont{10}{12.0}{\familydefault}{\mddefault}{\updefault} 0.1}}}}
\put(1575,-262){\makebox(0,0)[rb]{\smash{{\SetFigFont{10}{12.0}{\familydefault}{\mddefault}{\updefault} 1}}}}
\put(1646,-2559){\makebox(0,0)[b]{\smash{{\SetFigFont{10}{12.0}{\familydefault}{\mddefault}{\updefault} 1}}}}
\put(3151,-2559){\makebox(0,0)[b]{\smash{{\SetFigFont{10}{12.0}{\familydefault}{\mddefault}{\updefault} 10}}}}
\put(1189,-1292){\rotatebox{90.0}{\makebox(0,0)[b]{\smash{{\SetFigFont{10}{12.0}{\familydefault}{\mddefault}{\updefault}Extinction probability}}}}}
\put(2984,-2737){\makebox(0,0)[b]{\smash{{\SetFigFont{10}{12.0}{\familydefault}{\mddefault}{\updefault}$r$}}}}
\put(4005,-389){\makebox(0,0)[b]{\smash{{\SetFigFont{10}{12.0}{\familydefault}{\mddefault}{\updefault}$k=6$}}}}
\end{picture}%
\\
\begin{picture}(0,0)%
\includegraphics{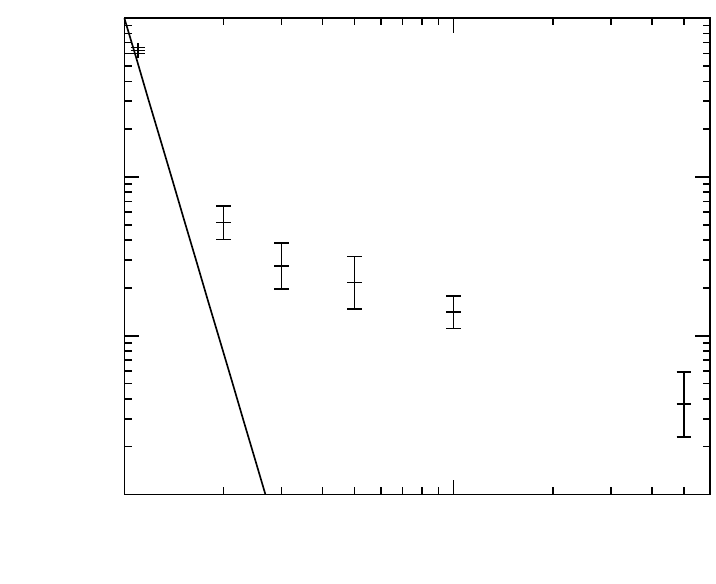}%
\end{picture}%
\setlength{\unitlength}{4144sp}%
\begingroup\makeatletter\ifx\SetFigFont\undefined%
\gdef\SetFigFont#1#2#3#4#5{%
  \reset@font\fontsize{#1}{#2pt}%
  \fontfamily{#3}\fontseries{#4}\fontshape{#5}%
  \selectfont}%
\fi\endgroup%
\begin{picture}(3257,2659)(1078,-2787)
\put(1575,-2440){\makebox(0,0)[rb]{\smash{{\SetFigFont{10}{12.0}{\familydefault}{\mddefault}{\updefault} 0.001}}}}
\put(1575,-1714){\makebox(0,0)[rb]{\smash{{\SetFigFont{10}{12.0}{\familydefault}{\mddefault}{\updefault} 0.01}}}}
\put(1575,-988){\makebox(0,0)[rb]{\smash{{\SetFigFont{10}{12.0}{\familydefault}{\mddefault}{\updefault} 0.1}}}}
\put(1575,-262){\makebox(0,0)[rb]{\smash{{\SetFigFont{10}{12.0}{\familydefault}{\mddefault}{\updefault} 1}}}}
\put(1646,-2559){\makebox(0,0)[b]{\smash{{\SetFigFont{10}{12.0}{\familydefault}{\mddefault}{\updefault} 1}}}}
\put(3151,-2559){\makebox(0,0)[b]{\smash{{\SetFigFont{10}{12.0}{\familydefault}{\mddefault}{\updefault} 10}}}}
\put(1189,-1292){\rotatebox{90.0}{\makebox(0,0)[b]{\smash{{\SetFigFont{10}{12.0}{\familydefault}{\mddefault}{\updefault}Extinction probability}}}}}
\put(2984,-2737){\makebox(0,0)[b]{\smash{{\SetFigFont{10}{12.0}{\familydefault}{\mddefault}{\updefault}$r$}}}}
\put(4005,-389){\makebox(0,0)[b]{\smash{{\SetFigFont{10}{12.0}{\familydefault}{\mddefault}{\updefault}$k=7$}}}}
\end{picture}%
&
\begin{picture}(0,0)%
\includegraphics{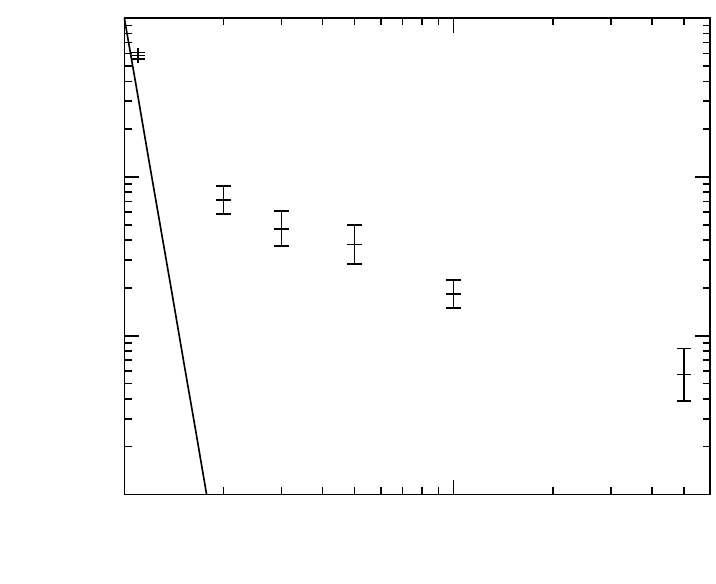}%
\end{picture}%
\setlength{\unitlength}{4144sp}%
\begingroup\makeatletter\ifx\SetFigFont\undefined%
\gdef\SetFigFont#1#2#3#4#5{%
  \reset@font\fontsize{#1}{#2pt}%
  \fontfamily{#3}\fontseries{#4}\fontshape{#5}%
  \selectfont}%
\fi\endgroup%
\begin{picture}(3257,2659)(1078,-2787)
\put(1575,-2440){\makebox(0,0)[rb]{\smash{{\SetFigFont{10}{12.0}{\familydefault}{\mddefault}{\updefault} 0.001}}}}
\put(1575,-1714){\makebox(0,0)[rb]{\smash{{\SetFigFont{10}{12.0}{\familydefault}{\mddefault}{\updefault} 0.01}}}}
\put(1575,-988){\makebox(0,0)[rb]{\smash{{\SetFigFont{10}{12.0}{\familydefault}{\mddefault}{\updefault} 0.1}}}}
\put(1575,-262){\makebox(0,0)[rb]{\smash{{\SetFigFont{10}{12.0}{\familydefault}{\mddefault}{\updefault} 1}}}}
\put(1646,-2559){\makebox(0,0)[b]{\smash{{\SetFigFont{10}{12.0}{\familydefault}{\mddefault}{\updefault} 1}}}}
\put(3151,-2559){\makebox(0,0)[b]{\smash{{\SetFigFont{10}{12.0}{\familydefault}{\mddefault}{\updefault} 10}}}}
\put(1189,-1292){\rotatebox{90.0}{\makebox(0,0)[b]{\smash{{\SetFigFont{10}{12.0}{\familydefault}{\mddefault}{\updefault}Extinction probability}}}}}
\put(2984,-2737){\makebox(0,0)[b]{\smash{{\SetFigFont{10}{12.0}{\familydefault}{\mddefault}{\updefault}$r$}}}}
\put(4005,-389){\makebox(0,0)[b]{\smash{{\SetFigFont{10}{12.0}{\familydefault}{\mddefault}{\updefault}$k=12$}}}}
\end{picture}%
\end{tabular}
\caption{Extinction probabilities for superstars with $\ell=m=200$ and
    $k$ as shown.  The straight line is $r^{-k}$ and the data points
    are the simulated probabilities.  The error bars indicate 99.5\%
    confidence intervals and $r^{-k}$ falls outside the confidence
    interval in every case apart from the three points marked ${}^*$.}
\label{fig:Superstar-sim}
\end{figure}

For clarity, we have plotted extinction probability (i.e., $1-f(G;r)$)
rather than fixation probability, and we have plotted on a log--log
scale.  The straight line shows the value of $r^{-k}\!$, i.e., the
extinction probability  predicted by Proposition~\ref{thm:clarified} and
the points are the fixation probabilities
derived by simulation, along with their 99.5\% confidence
intervals.\footnote{Brown, Cai and DasGupta~\cite{BCD2001:BinInt} and others
  have shown that the standard (Wald) binomial confidence interval of
  $p\pm z_{\alpha/2} \sqrt{p(1-p)/n}$ has severely chaotic behaviour,
  especially when $p$ is close to 0 or~1, as here, even for values of
  $n$ in the thousands.  This unpredictably produces confidence
  intervals with much lower coverage probabilities than the nominal
  confidence level --- often by 10\% or more.  Following the
  discussion in~\cite{BCD2001:BinInt}, we use what they call the
  Agresti--Coull interval, which applies a small adjustment to $p$ and
  $n$ before computing the interval.  This avoids the erratic
  behaviour of the Wald interval and gives coverage probabilities that
  are closer to the nominal confidence level and generally exceed it
  for $p$ close to 0 or~1.}
The only parameter values that we simulated for which $r^{-k}$ falls
within the 99.5\%
confidence interval of our simulations are $k=3$ and $r\in\{1.1, 2\}$
and $k=4$, $r=1.1$; these are the points marked~${}^*$ in
Figure~\ref{fig:Superstar-sim}.  In all other cases, the extinction
probabilities are significantly higher than the claimed value of
$r^{-k}\!$, with the disparity growing as $k$
increases.\footnote{Quantitatively, these results would only be weakened
    slightly by using the standard Wald interval: $1-r^{-k}$ would be
    within the confidence interval for the additional points $k=3$,
    $r=3$ and $k=4$, $r=2$.}

Reading down the columns of Table~\ref{table:Superstar-sim}, it
can be seen that for $r\geq 3$, the fixation probabilities do not
increase towards~1 but tail off for larger values of $k$.  In
particular, the lower end of the 99.5\% confidence interval for $k=5$
is greater than the upper end of the corresponding interval for $k=12$
for $r\in\{3,5,10\}$.  This observation does suggest that the claimed
fixation probability in Proposition~\ref{thm:clarified} may be qualitatively wrong
in the sense that the fixation probability might not tend to~$1$ as $k$~increases.
However, we are inclined to believe that the proposition is
qualitatively correct, and that the tailing off in the data is explained
by the fact that, for large values of~$k$, the values of $\ell$ and $m$ which we
were able to simulate may may have been too small for the limiting behaviour to be apparent.

One can also consider the degenerate case $k=2$, which has chains of length
zero (i.e., direct edges) from the reservoir vertices to the centre:
that is, the superstar $S_{\ell,m}^2$ is just the complete bipartite
graph $K_{1,\ell m}$, also known as a ``star''.  Large stars have
fixation probability tending towards
$1-r^{-2}$ (see, for example, \cite{Broom}) which is 0.9996 for $r=50$.  This is
above the upper end of the 99.5\% confidence interval of all our
$r=50$ superstar simulations, but again we suspect that $\ell=m=200$ is   too small for our
simulations to exhibit limiting behaviour in that case.

Note that each graph in Figure~\ref{fig:Superstar-sim} corresponds to a
row of the table.  For fixed~$k$, the fixation probability does indeed
tend to~1 as $r$ increases and this is easily seen to hold for any
strongly connected graph.

Lieberman \emph{et al.}\@ simulated only the case $r=1.1$
with $k=3$ and $k=4$, on graphs of around 10,000 vertices (they do not
state what values of $\ell$ and $m$ they used).  Their results in
these cases are consistent with ours: they measure fixation
probabilities of approximately 0.25 and 0.30 for $k=3$ and $k=4$,
respectively.  For $r$ close to~1 and small~$k$, the fixation
probability is reasonably close to $1-r^{-k}\!$.

The reader is referred to the ancillary files for the simulation code,
a description of it and a proof of its correctness.  As Barbosa
\emph{et al.}\@ point out \cite{Barbosa2}, it is difficult to simulate
on large graphs because of resource constraints.  We use various
time-saving tricks that they discuss such as skipping simulation steps
where nothing changes \cite{Gil1977:ExactSimulation}.  We also
describe several optimizations that we use that are specific to
superstars.

\bibliographystyle{plain}
\bibliography{superstarfixprob}

\newpage
\appendix

\section{Mathematica code}
\label{app:Mathematica}

Here is the text of the Mathematica program that we ran to solve the
linear system.  We explain the code in the proof of
Theorem~\ref{thm:kfive}.

\begin{sloppypar}
\small
\begin{verbatim}

XonO := r;
XoffO := M;
OonP := r;
OoffP := 1;
PonQ := r;
PoffQ := 1;
QonV := r;
QoffV := L;
Vgo := r;
VoffX := 1/ (L * M);
otherXoffO := M - 1;
otherQoffV := L - 1;

DX := VoffX + XonO;
EQX := FX == XonO * FXO/DX;
DO := OonP + XoffO;
EQO := FO == OonP * FOP/DO;
DP := PonQ + OoffP;
EQP := FP == PonQ * FPQ/DP;
DQ := QonV + PoffQ;
EQQ := FQ == QonV * FQV/DQ;
DV := Vgo + QoffV;
EQV := FV == Vgo/DV;

DXO := OonP + VoffX + otherXoffO;
EQXO := FXO == (OonP * FXOP + VoffX * FO + otherXoffO* FX)/DXO;
DOP := PonQ + XoffO;
EQOP := FOP == (PonQ * FOPQ + XoffO * FP)/DOP;
DPQ := QonV + OoffP;
EQPQ := FPQ == (QonV * FPQV + OoffP * FQ)/DPQ;
DQV := PoffQ + otherQoffV + Vgo;
EQQV := FQV == (PoffQ * FV + otherQoffV * FQ + Vgo)/DQV;
DVX := QoffV + XonO + Vgo; 
EQVX := FVX == (QoffV * FX + XonO * FVXO + Vgo)/DVX;

DXP := VoffX + XonO + OoffP + PonQ;
EQXP := FXP == (VoffX * FP + XonO * FXOP + OoffP * FX + PonQ * FPQX)/
    DXP; 
DOQ := XoffO + OonP + PoffQ + QonV;
EQOQ := FOQ == (XoffO * FQ + OonP * FOPQ + PoffQ * FO + QonV * FQVO)/
    DOQ;
DPV := OoffP + PonQ + QoffV + Vgo;
EQPV := FPV == (OoffP * FV + PonQ * FPQV + QoffV * FP + Vgo)/DPV;
DQX := PoffQ + QonV + VoffX + XonO;
EQQX := FQX == (PoffQ * FX + QonV * FQVX + VoffX * FQ + 
      XonO * FXOQ)/DQX;
DVO := QoffV + Vgo + XoffO + OonP;
EQVO := FVO == (QoffV * FO + Vgo + XoffO * FV + OonP * FOPV)/DVO;

DXOP := VoffX + otherXoffO + PonQ;
EQXOP := FXOP == (VoffX *FOP + otherXoffO * FXP + PonQ * FNV)/DXOP;
DOPQ := XoffO + QonV;
EQOPQ := FOPQ == (XoffO * FPQ + QonV * FNX)/DOPQ;
DPQV := OoffP + otherQoffV + Vgo;
EQPQV := FPQV == (OoffP *FQV + otherQoffV * FPQ + Vgo)/DPQV;
DQVX := PoffQ + otherQoffV + Vgo + XonO;
EQQVX := FQVX == (PoffQ *FVX + otherQoffV * FQX + Vgo + XonO * FNP)/
    DQVX;
DVXO := QoffV + Vgo + otherXoffO + OonP;
EQVXO := FVXO == (QoffV * FXO + Vgo + otherXoffO * FVX + OonP * FNQ)/
    DVXO;

DXOQ := VoffX + otherXoffO + OonP + PoffQ + QonV;
EQXOQ := FXOQ == ( 
     VoffX * FOQ + otherXoffO * FQX + OonP * FNV + PoffQ * FXO + 
      QonV * FNP)/DXOQ;
DOPV := XoffO + PonQ + QoffV + Vgo;
EQOPV := FOPV == (XoffO * FPV + PonQ * FNX + QoffV * FOP + Vgo)/DOPV;
DPQX := OoffP + QonV + VoffX + XonO;
EQPQX := FPQX == (OoffP * FQX + QonV * FNO + VoffX * FPQ + 
      XonO * FNV)/DPQX;
DQVO := PoffQ + otherQoffV + Vgo + XoffO + OonP;
EQQVO := FQVO == (PoffQ * FVO + otherQoffV * FOQ + Vgo + 
      XoffO * FQV + OonP * FNX)/DQVO;
DVXP := QoffV + Vgo + XonO + OoffP + PonQ;
EQVXP := FVXP == (QoffV * FXP + Vgo + XonO * FNQ + OoffP * FVX + 
      PonQ * FNO)/DVXP;

DNX := Vgo + XoffO + otherQoffV;
EQNX := FNX == (Vgo + XoffO * FPQV + otherQoffV * FOPQ)/DNX; 
DNO := Vgo + XonO + OoffP + otherQoffV;
EQNO := FNO == (Vgo + XonO * Fall + OoffP *FQVX + otherQoffV * FPQX)/
    DNO; 
DNP := otherQoffV + otherXoffO + OonP + PoffQ + Vgo;
EQNP := FNP == (otherQoffV * FXOQ + otherXoffO * FQVX + OonP * Fall + 
      PoffQ *FVXO + Vgo)/DNP;
DNQ := QoffV + otherXoffO + PonQ + Vgo;
EQNQ := FNQ == (QoffV * FXOP + otherXoffO * FVXP + PonQ * Fall + Vgo)/
    DNQ;
DNV := QonV + VoffX + otherXoffO;
EQNV := FNV == (QonV * Fall + VoffX *FOPQ + otherXoffO * FPQX)/DNV;

Dall := otherQoffV + otherXoffO + Vgo;
EQall := Fall == (otherQoffV * FNV + otherXoffO * FNO + Vgo)/Dall;

AllEQs := {EQX, EQO, EQP, EQQ, EQV, EQXO, EQOP, EQPQ, EQQV, EQVX, 
   EQXP, EQOQ, EQPV, EQQX, EQVO, EQXOP, EQOPQ, EQPQV, EQQVX,
   EQVXO, EQXOQ, EQOPV, EQPQX, EQQVO, EQVXP, EQNX, EQNO, EQNP,
   EQNQ, EQNV, EQall};
Allvars := {FX, FO, FP, FQ, FV, FXO, FOP, FPQ, FQV, FVX, FXP, FOQ, 
   FPV, FQX, FVO, FXOP, FOPQ, FPQV, FQVX, FVXO, FXOQ, FOPV, FPQX,
   FQVO, FVXP, FNX, FNO, FNP, FNQ, FNV, Fall};
 
SystemSolution := Solve[AllEQs, Allvars];
SolvedVars = Map[First, Part[SystemSolution, 1] ];
FXPos = Part[Part[Position[SolvedVars, FX] , 1], 1];
TheSolution := Part[Part[SystemSolution, 1], FXPos]

Soln := Collect[Collect[Simplify[Part[TheSolution, 2]], M], L];
(* The expression for FX as a function of L, M and r *)

MyNum := Numerator[Factor[Soln]];
MyDen := Denominator[Factor[Soln]];
Print["Numerator of FX"];
MonomialList[MyNum, {L, M}]
Print["Denominator of FX"];
MonomialList[MyDen, {L, M}]

\end{verbatim}
\end{sloppypar}

\end{document}